\documentclass[sigconf,nonacm]{acmart}
\usepackage{soul,balance}
\usepackage{booktabs} 
\usepackage{comment}
\usepackage{hyperref} 
\usepackage{graphicx} 	
\usepackage{float,wrapfig} 		    
\usepackage{mathtools,verbatim}      
\usepackage{caption}
\usepackage{subcaption}
\usepackage{multicol}
\usepackage{enumerate}
\usepackage{enumitem}
\newtheorem{observation}{Observation}

\usepackage[ruled, vlined,linesnumbered]{algorithm2e}

\begin{document}

\title{
Byzantine Causal Reliable Broadcast with Constant Metadata Overhead 
}

\author{Purv Patel}
\affiliation{%
  \institution{University of Illinois Chicago}
  \city{Chicago}
  \country{USA}
  }
\email{purvp2@uic.edu}

\author{Ajay D. Kshemkalyani}
\affiliation{%
  \institution{University of Illinois Chicago}
  \city{Chicago}
  \country{USA}
  }
\email{ajay@uic.edu}



\begin{abstract}
Asynchronous Byzantine Reliable Broadcast (BRB) is a fundamental primitive that guarantees agreement and validity in distributed systems subject to Byzantine faults, but it lacks ordering guarantees. Causal message ordering is important for many applications such as blockchain and social networking. Existing solutions for Byzantine Causal Reliable Broadcast (BCRB) have several drawbacks. Such protocols typically append vector clocks or dependency barriers to application messages, resulting in a metadata overhead that scales linearly with $n$, the number of processes in the system. 
In this paper, we propose the first optimal message overhead BCRB algorithm that guarantees safety. 
We do this by re-engineering Bracha's BRB algorithm with relatively small but critical modifications, and  prove that our algorithm solves BCRB with optimal message metadata.
The algorithm achieves constant-size $\mathcal{O}(1)$ message metadata overhead and $\mathcal{O}(n^2)$ messages, resulting in $\mathcal{O}(n^2)$ communication word complexity. This is as against $\mathcal{O}(n^3)$ communication word complexity of existing protocols. The algorithm tolerates $f < n/3$ Byzantine processes, which is the well-known optimal resilience bound, and uses four phases.
\end{abstract}


\maketitle

\section{Introduction}
Ensuring consistent transaction order in distributed systems subject to Byzantine failures is essential for applications ranging from decentralized state-machine replication (SMR) to financial ledgers and social networks. In these environments, Byzantine processes can observe pending messages, manipulate their delivery order, or inject front-running transactions to extract economic value or compromise consistency. 
While total-order primitives (atomic broadcast) are commonly used to resolve ordering conflicts, they require heavy consensus overhead and use randomization \cite{Cachin2001}. 
Causal ordering offers an alternative ordering guarantee by delivering messages according to Lamport's \textit{happens-before} relationship ($\to$) \cite{LLclock}. Causal ordering (CO) of broadcasts under crash failures has been studied extensively since Birman and Joseph's seminal paper about 40 years ago \cite{birman1987reliable}. On the other hand, Bracha defined Byzantine Reliable Broadcast (BRB) and proposed an elegant algorithm to solve it, also about 40 years ago \cite{DBLP:journals/iandc/Bracha87}. BRB has also been extensively used and studied since then due to its importance. However, combining the two (BRB + CO) {\em efficiently} to solve Byzantine Causal Reliable Broadcast (BCRB) has remained an open problem since then. 
%
Scaling BCRB to large networks is severely limited by previous algorithms which append $\mathcal{O}(n)$ metadata to messages and incur $\mathcal{O}(n^3)$ communication word complexity, where $n$ is the number of processes 
\cite{DBLP:journals/tcs/AuvolatFRT21,Cachin2001}. The only known BCRB algorithm with $\mathcal{O}(1)$ metadata overhead, $\mathcal{O}(n^2)$ messages, and $\mathcal{O}(n^2)$ communication word complexity has provable bounds on the probability of violation of safety \cite{PKbcrb}. Our result fills in the gap by also guaranteeing safety.

Decoupling causal delivery from the network layer while maintaining a constant $\mathcal{O}(1)$ message metadata overhead presents fundamental challenges in asynchronous systems. Bracha proposed the first BRB protocol \cite{DBLP:journals/iandc/Bracha87}. It has been formally shown by Misra et al. \cite{misra2022causal} that Bracha's protocol does not satisfy even a weakened form of safety, called weak safety, in the presence of even one Byzantine process. It has also been proved that (deterministically) achieving both strong causal safety and liveness without using cryptography is impossible in asynchronous Byzantine systems \cite{misra2022detecting, DBLP:journals/tpds/MisraK24,DBLP:journals/pc/MisraK25,DBLP:conf/netys/MisraK22,DBLP:conf/icdcn/MisraK23}. Existing  protocols for BCRB weaken safety to weak safety \cite{DBLP:journals/tcs/AuvolatFRT21,DBLP:journals/tpds/MisraK24}, or accept probabilistic safety guarantees as in \cite{PKbcrb,Cachin2001}. Such solutions either couple the application data directly with a costly consensus plane to prevent front-running \cite{Cachin2001}, accept front-running \cite{DBLP:journals/tcs/AuvolatFRT21,PKbcrb,minicast}, or incur linear metadata overheads that degrade throughput \cite{minicast,DBLP:journals/tcs/AuvolatFRT21, Cachin2001}.

\medskip
\textbf{Contributions:}
To resolve these bottlenecks, we present a novel BCRB protocol that achieves constant-size $\mathcal{O}(1)$ message metadata and $\mathcal{O}(n^2)$ messages.
We do this by re-engineering Bracha's BRB algorithm \cite{DBLP:journals/iandc/Bracha87} with relatively small but critical modifications that integrate the causal ordering logic with the BRB logic, and  prove that our algorithm solves BCRB, including strong safety and liveness, with optimal message metadata.
The algorithm also achieves $\mathcal{O}(n^2)$ communication word complexity as against $\mathcal{O}(n^3)$ communication word complexity of previous protocols \cite{Cachin2001,DBLP:journals/tcs/AuvolatFRT21}, which had at least $\mathcal{O}(n)$ message metadata. 
Our algorithm can tolerate $f < n/3$ Byzantine processes which is the well-known optimal fault-tolerance bound and uses four phases. 
Additionally, before proposing our algorithm, we also give a new definition of Byzantine happens-before and BCRB that attempts to characterize the problem precisely. 

Our algorithm is {\em throughput-scalable}, i.e., throughput and rate of sending are not limited by the latency of messages. 
A process does not wait for the previous BCRB instance to be locally delivered before sending in the next BCRB instance.

{\bf Outline:} Section~\ref{sec:background} gives the system model and background. Section~\ref{sec:algo} gives our algorithm (Algorithm~\ref{alg:brachacb}). Section~\ref{sec:proof} gives the correctness proof. Section~\ref{sec:concl} concludes. 

\section{System Model and Background}
\label{sec:background}

\subsection{System Model}
We assume an asynchronous distributed system of $n$ processes, denoted by $p_1, p_2, \dots, p_n$. The system is subject to Byzantine failures. Any $f$ processes can behave arbitrarily by dropping messages, forging causal histories, or colluding. The number of such Byzantine processes is bounded by $f < n/3$ (the BRB resilience bound which is optimal) \cite{DBLP:journals/iandc/Bracha87,DBLP:journals/jacm/BrachaT85}.  Processes communicate via reliable and authenticated point-to-point FIFO channels. Traditionally, the Byzantine Reliable Broadcast (BRB) primitive (e.g., Bracha's BRB \cite{DBLP:journals/iandc/Bracha87}) has been defined to satisfy Validity, Agreement, and Integrity.
\begin{definition}[Traditional Validity, Agreement, Integrity (VAI)]
The traditional definitions are as follows \cite{DBLP:journals/iandc/Bracha87}.
\begin{itemize}
     \item \textbf{Validity}: If a correct process broadcasts $m$, all correct processes eventually deliver $m$.
     \item \textbf{Agreement}: If a correct process delivers $m$, all correct processes eventually deliver $m$.
     \item \textbf{Integrity}: A message $m$ is delivered at most once by each correct process, and if the sender is correct, then only if it was broadcasted by the sender.
 \end{itemize}
\label{def:tradvai}
\end{definition}

Lamport's happens-before relation $\to$ \cite{LLclock} over the set of events in the system is as follows. For any two events $e_1, e_2$, we have $e_1 \to e_2$ if: (1) they occur at the same process and $e_1$ precedes $e_2$ in local execution; (2) $e_1$ is the broadcast event of a message and $e_2$ is the delivery event of that message; or (3) there exists an event $e_3$ such that $e_1 \to e_3$ and $e_3 \to e_2$.

For our protocol, we assume an $(n-f, n)$ threshold decryption scheme (e.g., Shoup's threshold cryptosystem \cite{shoup2000practical}), where a ciphertext can only be decrypted once at least $n-f$ valid decryption shares from distinct processes are collected.


\subsection{Broadcast Primitives}
\label{sec:bcprimitives}

We define two primitives of our architecture: 
    \begin{itemize}
        \item $\text{\texttt{bcrb\_broadcast}}(payload)$: Invoked by the application layer at process $p_i$ to causally broadcast a plaintext $payload$.
        \item $\text{\texttt{bcrb\_deliver}}(sender, sn, plaintext)$: Upcall from the BCRB layer to the application layer to deliver the ordered $plaintext$ payload originally sent by $p_{sender}$ with sequence number $sn$.
    \end{itemize}

We also identify three special events for a message received.
\begin{definition}[Special event definitions]
Three special events at a process per application message received are defined as follows.
\begin{enumerate}
\item $plaintext\_ready$: The plaintext of the payload of a received message is readable by the application. 
\item $bcrb\_ready$: The received message has cleared the Byzantine hurdles 
of ordering with respect to its causal predecessors; hence it is ready to be given to the application next, (pending any revelation required, if any).
\item $bcrb\_deliver$: The plaintext of a received message is delivered to the application, satisfying VAI requirements and ordering with respect to causal predecessors.
\end{enumerate}
\end{definition}
For any message and depending on the algorithm, at a process, $plaintext\_ready$ and $bcrb\_ready$ may occur in any order; $bcrb\_deliver$ occurs last. For the same message at a different process, the order of occurrence may be different.

\begin{definition}[Front-running attack]
A {\em front-running attack} on message $m_1$ is mountable by message $m_2$ sent by $p_i$ when $m_2$ is $\texttt{bcrb\_broadcast}$ between $plaintext\_ready(m_1)$ and $bcrb\_deliver(m_1)$ at $p_i$. 
The attack is successful if $\texttt{bcrb\_deliver}(m_2)$ occurs before $\texttt{bcrb\_\-de\-l\-iv\-er}(m_1)$ at a correct $p_j$.
\end{definition}

\subsection{Safety and Liveness Definitions}

To define safety, we first need a definition of the ``happens before''/causality relation $\to$ in Byzantine systems. 
Existing definitions have some drawbacks.  That in \cite{DBLP:journals/tcs/AuvolatFRT21} does not consider the  impact of Byzantine processes. That in \cite{DBLP:journals/tpds/MisraK24}, though correct for unicast mode, seems underspecified for broadcast because messages may not be delivered in the default sequential order of sending by a Byzantine process in an implementation due to lack of a FIFO Byzantine reliable broadcast channel. 
We propose a new definition that attempts to characterize the relation $\to$ on the set $\mathcal{M}$ of all messages delivered by correct processes, using only local events.

The first rule is based on that of Auvolat et al. \cite{DBLP:journals/tcs/AuvolatFRT21}. But we add that the total order on the set of messages sent by $p_i$ that are bcrb-delivered by correct processes is decided by the order of delivery, which is decided by the order on the sequence numbers of those messages. 
This is to simplify the implementation. This rule orders all the messages in $\mathcal{M}$ sent by a process on a per sender (process) basis.
The second rule captures the causality relation between a message $m$ received and a message $m'$ subsequently sent at any (correct or Byzantine) process. A causal dependency of $m'$ on $m$ may arise {\em if and only if} plaintext of $m$ is readable (at a $plaintext\_ready(m)$ event) before sending $m'$ (a \texttt{bcrb\_broadcast}($m'$) event). 
The third rule captures transitivity of $\to$ to account for the induced ordering introduced by message chains.  
With this definition of $\to$, specifically due to the second rule, front-running attacks by Byzantine processes are ruled out in any system that implements causal safety based on this new definition.

\begin{definition}[Happens before relation: New definition for Byzantine systems]
\label{def:hbbyz}
Let $\mathcal{M}$ be the set of (application) messages bcrb-delivered by correct processes in an execution. A Byzantine causal order on $\mathcal{M}$ is a partial order $\to$  on $\mathcal{M}$, i.e., for $m, m' \in \mathcal{M}$, $m \to m'$ if any of these rules is satisfied:
\begin{enumerate}
\item For any (correct or faulty) process $p_i$, the set of messages from $p_i$ bcrb-delivered by any correct process $p_j$ is totally ordered by $\to$, defined as the order of delivery of those messages. This order of delivery can be chosen to be the order on the sequence numbers of those messages.
\item If $p_i$ \texttt{plaintext\_ready}($m$) before it \texttt{bcrb\_broadcast}($m'$), then $m \to m''$. 
\item If $m \to m'$ and $m' \to m''$, then $m \to m''$.
\end{enumerate}
\end{definition}

This new definition captures the true intent of ``happens before'' or  causality, viz. (potential) influence between messages while accounting for such influences introduced by Byzantine processes also. 
This ``happens before''/causal order relation can be seen to be a partial order.

In contrast, an alternate definition for Byzantine systems using the $bcrb\_ready$ event instead of the $plaintext\_ready$ event in the second rule is not adequate because the plaintext of a message may not be available at this event (the definition would be stronger) or may be available earlier (the definition would be weaker). An example of the latter is a BRB algorithm (such as Bracha's BRB) that uses no encryption.

To make Definition~\ref{def:hbbyz} applicable to any general communication (not just reliable broadcast) in a Byzantine system, replace ``bcrb\_broadcast'' and ``bcrb-delivered'' by ``sent'' and ``delivered'', respectively. The resulting definition also addresses the drawback of the definition in \cite{DBLP:journals/tpds/MisraK24}.

\begin{definition}[Causal Safety and Liveness]
\label{def:sl}
Safety and Liveness are defined as follows, where the $\to$  relation is defined as per Definition~\ref{def:hbbyz}:
\begin{itemize}
    \item \textbf{(Strong) Causal Safety}: If $m_1 \to m_2$ then no correct process triggers \texttt{bcrb\_deliver}($m_2$) before \texttt{bcrb\_deliver}($m_1$). 
    \item \textbf{Weak Causal Safety}: If $m_1$ $\to$ $m_2$ and the causal chain from $m_1$ to $m_2$ passes exclusively through correct processes, then no correct process triggers \texttt{bcrb\_deliver}($m_2$) before \texttt{bcrb\_deliver}($m_1$) \cite{DBLP:journals/tpds/MisraK24,DBLP:conf/netys/MisraK22,DBLP:conf/icdcn/MisraK23}.
    \item \textbf{Liveness}: Every message broadcasted via \texttt{bcrb\_broadcast}\-(payload) by a correct process is eventually delivered via \texttt{bcrb\_\-deliver}(payload) at all correct processes. 
    
    Liveness is same as Validity.
\end{itemize}
\end{definition}

\begin{table*}[htbp]
\centering
\caption{Comparative Analysis of Byzantine Causal Broadcast Protocols}
\label{tab:comparison}
{\footnotesize
\begin{tabular}{lccccccc}
\hline
\textbf{Property} & \textbf{Auvolat et al. \cite{DBLP:journals/tcs/AuvolatFRT21}} & \textbf{Cachin et al. \cite{Cachin2001}} & 
\textbf{\cite{PKbcrb}, Alg. 1} & 
\textbf{\cite{PKbcrb}, Alg. 2} & 
\textbf{\cite{PKbcrb}, Alg. 3} & 
\textbf{\cite{PKbcrb}, Alg. 4} & 
Proposed \\
\hline
Deterministic & Yes & No (randomized) & Yes & Yes & Yes & Yes & Yes
\\
Crypto & No & Yes & Yes & No & Yes & No & Yes
\\
Message Overhead Size & $\mathcal{O}(n)$ 
 & $\mathcal{O}(n)$ (Encrypt. Hdrs) 
 & $\mathcal{O}(1)$  & $\mathcal{O}(1)$  & $\mathcal{O}(1)$ & $\mathcal{O}(1)$ & $\mathcal{O}(1)$
 \\
Communic. Complexity & $\mathcal{O}(n^3)$ & $\mathcal{O}(n^3)$ & $\mathcal{O}(n^2)$ & $\mathcal{O}(n^2)$ & $\mathcal{O}(n^3)$ & $\mathcal{O}(n^3)$ & $\mathcal{O}(n^2)$ 
\\
Front-Running Protection & No & Yes & Yes & No & Yes & No & Yes
\\
Throughput-scalable & Yes & No & Yes & Yes & Yes & Yes & Yes
\\
$P$(weak safety violation) & 0 & $\epsilon > 0$ & $\mathcal{O}(f^{-3} \cdot \ln^3 f)$ & $\mathcal{O}(f^{-3} \cdot \ln^3 f)$ & 0 & 0 & 0
\\
$P$(strong safety violation) & high; front-running $+$  & $\epsilon > 0$ & $\mathcal{O}(f^{-1} \cdot \ln^2 f)$ & front-running  & $\mathcal{O}(f^{-1} \cdot \ln^2 f)$ & front-running & 0
\\
 & fake causal barriers & & & & \\
\hline
\end{tabular}
}
\end{table*}

\textbf{Problem Definition.} BCRB must satisfy Validity, Agreement, Integrity, (Strong) Safety, and Liveness. 
Validity, Agreement, and Integrity is the traditional definition with the broadcast and deliver events replaced by $\texttt{bcrb\_broadcast}$ and $\texttt{bcrb\_deliver}$, as follows.
\begin{definition}[Validity, Agreement, Integrity (VAI)] 
The definitions for Byzantine systems are as follows.
\begin{itemize}
     \item \textbf{Validity}: If a correct process $\texttt{bcrb\_broadcast}$ $m$, all correct processes eventually $\texttt{bcrb\_deliver}$ $m$.
     \item \textbf{Agreement}: If a correct process $\texttt{bcrb\_deliver}$s $m$, all correct processes eventually $\texttt{bcrb\_deliver}$ $m$.
     \item \textbf{Integrity}: A message $m$ is $\texttt{bcrb\_deliver}$ed at most once by each correct process, and if the sender is correct, then only if it was $\texttt{bcrb\_broadcast}$ed by the sender.
 \end{itemize}
\label{def:tradvaibc}
\end{definition}

\subsection{Background: Byzantine Causal Broadcast}
Misra et al. \cite{DBLP:journals/tpds/MisraK24,DBLP:journals/pc/MisraK25,DBLP:conf/netys/MisraK22} showed that it is impossible to provide both strong safety and liveness for causal ordering without using cryptography, but weak safety and liveness can be provided. It was formally proved in \cite{misra2022causal} that Bracha's BRB \cite{DBLP:journals/iandc/Bracha87} does not satisfy even the weak safety property in the presence of even one Byzantine process.

Byzantine causal broadcast algorithm by Auvolat et al. \cite{DBLP:journals/tcs/AuvolatFRT21} enforces causal order by attaching a ``causal barrier" (predecessor message IDs) to every message. This incurs an $\mathcal{O}(n)$ space overhead in application messages and runs directly over an $\mathcal{O}(n^2)$ messages BRB primitive, generating $\mathcal{O}(n^3)$ message communication complexity. The algorithm provides liveness and weak safety.
Cachin et al. \cite{Cachin2001} proposed a secure causal atomic broadcast protocol using threshold encryption. A sender broadcasts the encrypted payload via an atomic broadcast channel. The use of atomic broadcast leads to a probabilistic, randomized (non-deterministic) solution. Once the ciphertext is totally ordered, processes exchange decryption shares to reveal the plaintext. Although this prevents front-running, it tightly couples the data and control planes, forcing large application payloads to be processed by the expensive total-order consensus layer ($\mathcal{O}(n^3)$ communication word complexity and $\mathcal{O}(n)$ message metadata overhead). This protocol provides liveness, and being randomized, provides weak safety and strong safety with high probability. 
Patel et al. \cite{PKbcrb} proposed an $\mathcal{O}(1)$ metadata overhead algorithm with three variants, all having $\mathcal{O}(1)$ message metadata. By modeling link propagation delays as independent exponential variables, they derived bounds on the probability of weak safety violations and of strong safety violations. 
%
Our work fills in the gap by guaranteeing (strong) safety.
A comparison of these protocols with our protocol is given in Table~\ref{tab:comparison}. 

\subsection{Background: Bracha's Reliable Broadcast}
To ground our algorithm and its correctness proof, we first  review Bracha's Byzantine Reliable Broadcast (BRB) algorithm \cite{DBLP:journals/iandc/Bracha87}, which our algorithm reengineers. Bracha's protocol uses three types of messages: \texttt{init} (for \texttt{initial}), \texttt{echo}, and \texttt{ready}. A broadcast is initiated by a sender broadcasting \texttt{init}. The correct processes transition through three steps based on threshold quorums:
\begin{itemize}
    \item \textbf{Step 1}: A process waits to receive one \texttt{init} message from the sender. 
    It then broadcasts \texttt{echo} to all.
    \item \textbf{Step 2}: A process waits to receive $\frac{n+f}{2}$ \texttt{echo} messages or $f+1$ \texttt{ready} messages. It then broadcasts \texttt{ready} to all.
    \item \textbf{Step 3}: A process waits to receive $2f+1$ \texttt{ready} messages, after which it delivers (accepts) the payload.
\end{itemize}
These step thresholds ensure that even if a Byzantine sender attempts to equivocate, no two correct processes can deliver different payloads, and if any correct process delivers a payload, all correct processes eventually deliver it. That is, Validity, Agreement, and Integrity are satisfied.

\section{Proposed Algorithm}
\label{sec:algo}

\begin{algorithm*}[th!]
\SetKwComment{Comment}{$\triangleright$ }{}

\begin{multicols}{2}
\SetInd{0.2em}{0.5em}
\textbf{type and state variables:} \\
  $seq_i \gets 0$ \Comment*{{\footnotesize Sequence number for process $i$'s broadcasts}}
  $\mathbf{LD}_i \gets \text{Array}[n] \text{ of } 0$ \Comment*{{\footnotesize Last msg IDs queued for BCRB delivery}}
  $\mathbf{D}_i \gets \text{Array}[n][n] \text{ of } 0$ \Comment*{{\footnotesize $D[x,y]$ is $LD_y[x]$ as reported by $y$}}   
  $echo\_set \gets \emptyset$ \Comment*{{\footnotesize received ECHO messages}}
  $ready\_set \gets \emptyset$ \Comment*{{\footnotesize received READY messages}}
  $shares[M] \gets \emptyset$ \Comment*{{\footnotesize shares received for $M ( \gets (p_j, seq_j, m))$}}
  $bcrb\_Q \gets [\,]$ \Comment*{{\footnotesize Buffered messages for BCRB delivery}}
  Msg: record \{ \\
    $from: \mathbb{N}$ \Comment*{{\footnotesize sender of $M$}}
    $blocked: \mathbb{B}$ \Comment*{{\footnotesize Flag: $\neg m.CB \le LD$?}}
    $CB: \text{Array}[n] \text{ of } \mathbb{N}$ \Comment*{{\footnotesize Causal barrier vector}}    
 \} \\

\BlankLine
\textbf{upon} the application is ready to BCRB broadcast message $payload$: \\
$seq_i \gets seq_i + 1$\\
$ciphertext \gets encrypt(payload, PK)$\\
broadcast INIT($p_i, seq_i, ciphertext$) \\

\BlankLine
\textbf{upon} INIT($p_j, seq_j, m$) is received: \\ 
broadcast ECHO($p_j, seq_j, m$)
\\

\BlankLine
\textbf{upon} ECHO($M \gets (p_j, seq_j, m)$) is received from $s$: \\
$\forall a, M.CB[a] \gets D[a,s]$\\ 
$M.from \gets s$\\  
$M.blocked \gets \neg (\forall a, D[a,s] \le LD[a])$\\ 
$echo\_set \gets echo\_set \cup \{M\}$\\

\BlankLine
\textbf{upon} READY($M \gets (p_j , seq_j ,m)$) is received from $s$:\\ 
$\forall a, M.CB[a] \gets D[a,s]$\\ 
$M.from \gets s$\\  
$M.blocked \gets \neg (\forall a, D[a,s] \le LD[a])$\\ 
$ready\_set \gets ready\_set \cup \{M\}$\\

\columnbreak
\BlankLine
\textbf{upon} $(M \gets (j, seq_j, m)) \in echo\_set$ received from more than $(n+f)/2$ different $M.from$ such that their $M.blocked=0$ $\land$ READY($p_j, seq_j, m$) not yet broadcast:\\
        broadcast READY($p_j, seq_j, m$)\\

\BlankLine
\textbf{upon} $(M \gets (p_j, seq_j, m)) \in ready\_set$ is received from $(f+1)$ different $M.from$ such that their $M.blocked=0$ $\land$ READY($p_j, seq_j, m$) not yet broadcast:\\
        broadcast READY($p_j, seq_j, m$)

\BlankLine
\textbf{upon} $(M \gets READY(p_j, seq_j, m)) \in ready\_set$ is received from $(2f+1)$ different $M.from$ $\land$ $p_i$ has broadcast READY($p_j, seq_j, m$) $\land$ $LD[j]=seq_j-1$: 
\label{condcb}\\
    $enqueue(bcrb\_Q,(j, seq_j, m).(flag\gets 0))$\\ $LD[j]\leftarrow seq_j$ \label{startingcb} \\ 
    $share \gets dec\_share(m,SK_i)$\\
    broadcast ACK($p_j, seq_j, m, share, i$)\\
    \For{$\text{ each } M' \in (echo\_set \cup ready\_set)$ such that $M'.blocked=1$}{
    $M'.blocked \gets is\_blocked(M')$
    }

\BlankLine
\textbf{upon} ACK($(M \gets (p_j, seq_j, m)), share, s$) is received $\land$ $share$ is valid: \\
$D[j,s]\leftarrow \max(D[j,s], seq_j)$\\
$shares[M] \gets shares[M] \cup \{share\}$\\

\BlankLine
\textbf{upon} ACK($(M \gets (p_j, seq_j, m))$), $share$, $s$) with valid $share$ is received from $(n-f)$ different $M.from$ $\land$ $|shares[M]| \ge n-f$ $\land$ $M.flag=0$ is in $bcrb\_Q$:\\
        $plaintext\gets decrypt(m, shares[M], VK)$ \label{line:ptr} \\
        replace $((p_j, seq_j, m).flag=0)$ in $bcrb\_Q$ by $((p_j, seq_j, plaintext).flag=1)$\\
        \While{$bcrb\_Q.head.flag=1$}{
            $\texttt{bcrb\_deliver}(dequeue(bcrb\_Q).head)$\\
        }        

\BlankLine
\textbf{function} $is\_blocked(M(p_j, seq_j, *))$:\\
return($\neg (\forall a \in [1,n], M.CB[a] \le LD[a])$)\\
\end{multicols}
\caption{Byzantine Reliable Causal Broadcast Protocol, $f \leq \lfloor (n-1)/3 \rfloor$. Code for $p_i$.}
\label{alg:brachacb}
\end{algorithm*}

\subsection{Basic Idea}
\label{sec:basicidea}
Rather than run the causal ordering logic in a layer above the BRB layer that provides the VAI properties, we integrate the two functionalities. Our solution routes payloads via $\mathcal{O}(n^2)$ BRB and then exchanges decryption shares. Specifically, we reengineer Bracha's algorithm with small but critical changes that provide (strong) safety with constant message overhead. 

In the Auvolat et al. solution \cite{DBLP:journals/tcs/AuvolatFRT21}, a $\mathcal{O}(n)$ sized {\em causal barrier}, capturing a BCRB-broadcast message's causal dependencies, is piggybacked on the message passed to the BRB protocol. On BRB-delivery, the causal ordering layer atop buffers the message until the causal predecessor messages that were identified by the causal barrier are delivered. Then the message is cleared for BCRB-delivery. 

In our algorithm (Algorithm~\ref{alg:brachacb}), we identify as {\em BCRB-ready event} for a message $m$, that event at which a correct process has successfully processed $2f+1$ READY($m$) messages after the third phase of Bracha's BRB. (This event will be defined formally in the Correctness proof section. From the proof of Theorem~\ref{th:safety} (Safety), it will become clear that this event satisfies the description given in Section~\ref{sec:bcprimitives}.)  On a BCRB-ready event, $m$ is enqueued for BCRB-delivery (pending possible decryption).  We also add a fourth {\em ACK} phase to Bracha's BRB, in which a process broadcasts the decryption share on an ACK message after the BCRB-ready event. On receiving $n-f$ valid ACKs, the BCRB-readyed message is decrypted. 
As safety is satisfied, causal predecessors would have been already enqueued ahead of $m$.

As for enforcing causal ordering, we push the causal barrier checks inside Bracha's BRB. Specifically, we exploit the following well-known property of Bracha's BRB, here modified by adding causal barrier checks:
\begin{itemize}
\item A READY is not broadcast by correct $p_i$ until the causal barriers associated with the broadcasts of the $\frac{n+f}{2}$ $p_k$ processes' ECHOs or associated with the broadcasts of the $f+1$ $p_k$  processes' READY that can trigger the broadcast of the READY at $p_i$ are satisfied locally, i.e., until those messages in the causal barriers are BCRB-ready. 
\end{itemize}
We show as part of the proof of Theorem~\ref{th:safety} that these causal barriers satisfy another essential property in our algorithm: 
\begin{itemize}
    \item If $\texttt{plaintext\_ready}(m_1) \rightarrow \texttt{bcrb\_broadcast}(m_2)$ then there must exist at least one correct process $p_k$ (as identified above) that broadcasts an ECHO($m_2$) or READY($m_2$) and the causal barrier imposed by that ECHO or READY includes $m_1$.
\end{itemize}
Thus any correct $p_i$ will not broadcast READY($m_2$) (and hence not BCRB-ready it and enqueue it for BCRB-delivery) until $m_1$ has been locally BCRB-readyed and enqueued for BCRB-delivery earlier. This guarantees (strong) safety.

Note that pushing the causal barrier checks inside Bracha's BRB does not alter its properties; it simply delays when a process transitions to the next phase (i.e., after the causal barrier of a received ECHO or READY is satisfied, for each ECHO or READY in Bracha's required quorum sizes). 

A critical question arises: how is the causal barrier per ECHO and READY message computed and transmitted? Our algorithm {\em implicitly, indirectly, dynamically}, and {\em incrementally} computes/maintains the causal barrier using smart data structures and leveraging the semantics of the ACK message (required for sending a decryption share), without actually transmitting any data or metadata for this end. 
\begin{itemize}
\item Each process $p_i$ maintains a $n \times n$ matrix $D$. When ACK($(p_a,\- seq_a, m))$ is received from $p_s$ with a valid share, it updates $D[a,s]$ $\gets$ $seq_a$. 
\item Process $p_i$ also uses a size $n$ vector $LD$ where $LD[k]$ $\gets$ $seq_k$ when it has received $2f+1$ READYs for  $(M \gets (k, seq_k, m))$ and the causal barriers on these READYs are satisfied and $M$ is consequently BCRB-readyed. 
\item Now, when an ECHO or READY for $(M' \gets (j, seq_j, m))$ is received from $p_s$, $p_i$ {\em creates} a $n$-vector causal barrier $CB$ for it, by setting $\forall a, M'.CB[a] \gets D[a,s]$ and storing it locally. $M'.blocked$ is set to ($\neg (\forall a \in [1,n], M'.CB[a] \le LD[a])$) and $M'.from \gets s$. (There will be multiple instances of $M'$ at $p_i$, each with a potentially different $M'.from$ value.) 
\item If $M'.blocked=1$, it is reevaluated whenever any entry in $LD$ is incremented. When the causal barrier is satisfied for the Bracha-required quorum of ECHO or READY messages for $M'$, $p_i$ broadcasts READY($M'$).
\end{itemize}
As our correctness proof implies, Byzantine processes cannot game this logic.



\section{Correctness Proof}
\label{sec:proof}
\subsection{Safety}

\begin{observation}[Impact of ACK on causal barrier]
\label{obs:ackDcbLD}
When an ACK is sent and when an ACK is received:
\begin{enumerate}
    \item When ACK($M \gets (j, seq_j, m)$) is sent by $p_z$, $LD_z[j] = seq_j$.
    \item When ACK($M \gets (j, seq_j, m)$) is received from $p_z$, $D[j,z] \gets seq_j$.
    \item Subsequently on receiving any ECHO/READY($M' (\gets (k, seq_k, m')$) from $p_z$, $M'.CB[j] \gets D[j,z] \ge seq_j$ and $M'.from \gets z$.
\end{enumerate}
\end{observation}

\begin{observation}[Impact of receiving ECHO/READY]
For ECHO\-/READY$((M \gets (p_j, seq_j, *)))$ that is received from correct $p_k$ by correct $p_i$, 
\begin{enumerate}
    \item $D_i[j,k]$ on receipt at $p_i$ equals $LD_k[j]$ at the time of sending of that ECHO/READY.
    \item At $p_i$, $\forall a \in [1,n], M.CB[a] = LD_k[a]$ at the time of sending that ECHO/READY.
    \item At $p_i$, when $M.blocked=0$, $\forall a \in [1,n], LD_i[a] \geq LD_k[a]$ at the time of sending $M$ at $p_k$.
\end{enumerate}
\label{obs:equality}
\end{observation}

We define a {\em BCRB-ready} event of $M=(j, seq_j, m)$ at $p_i$ as the event at which $p_i$ executes $LD_i[j]$ to $seq_j$ (line~\ref{startingcb}) in Algorithm~\ref{alg:brachacb}. Similarly, we define a $plaintext\_ready$ event  of $M=(j, seq_j, m)$ at $p_i$ as the event at which $p_i$ executes decryption (line~\ref{line:ptr}) in Algorithm~\ref{alg:brachacb}.

{\bf Henceforth,  we may omit the payload (third parameter) and subsequent parameters on algorithm messages for brevity.}

\begin{lemma}
\label{lemma:agree}
If a correct process $p_i$ BCRB-readys $M$, every correct process also (eventually) BCRB-readys $M$.
\end{lemma}
\begin{proof}
We use the interleaved event model (rather than the poset model) of the distributed execution on the set of BCRB-ready events {\em at correct processes}, i.e., we project such events on a common global time axis. Let the $x$th BCRB-ready event occur at global time $t_x$. We show the following Invariant by induction: 
\begin{itemize}
    \item (Invariant/Induction hypothesis:) if $M$ is BCRB-readyed at $t_x$ (at a correct process $p_i$), it will be BCRB-readyed at all correct processes $p_c$.
\end{itemize}
Process $p_i$ must have received $2f+1$ READY($M$), at least $f+1$ having been broadcast by correct $p_k$.
%
All correct $p_c$ receive $f+1$ READY($M$) broadcasts by the $p_k$, which must have done the broadcast before $t_x$. When $p_c$ receives such a broadcast:
\begin{itemize}
\item For Base case $x=1$, $LD_c[a] \ge LD_k[a]$ = 0 because $p_k$ did the broadcast before $t_1$. So $p_c$ will broadcast READY($M$). This happens at each of the $2f+1$ $p_c$s. On receipt of $2f+1$ READY($M$), each $p_c$ will BCRB-ready $M$. 
\item For $x>1$: Assume I.H. true for all $y<x$.  Eventually $LD_c[a] \geq LD_k[a] 
(=, \text{ say}, seq_a)
$ at the time of $p_k$'s broadcast which happened before $t_x$.  This is justified as follows. $LD_k[a] = seq_a$ means $M_a (\gets (a, seq_a))$ was BCRB-readyed at $p_k$ (per definition). Further, this happened at some $t_y$ which happened before $p_k$'s broadcast which happened before $t_x$. Hence $y<x$.  From I.H., $M_a$ is guaranteed to be eventually BCRB-readyed at each $p_c$.
At $p_c$, denote as $M_{k'}$ the READY($M$) received from $p_{k'}$. As eventually at $p_c$, for each of the $f+1$ $M_{k'}$, $\forall a, LD_c[a] \geq M_{k'}.CB[a] = LD_{k'}[a]$ at the time $p_{k'}$ broadcast $M_{k'}$, (the equality follows from Observation~\ref{obs:equality}), $p_c$ will broadcast READY($M$). This happens at each of the $2f+1$ $p_c$s.
In addition, on receipt of $2f+1$ READY($M$), each $p_c$ will BCRB-ready $M$. 
\end{itemize}
The lemma follows.
%
\end{proof}

\begin{theorem}[BCRB Safety]
If $m \to m'$
then no correct process $p_l$ \texttt{bcrb\_deliver}s $m'$ before \texttt{bcrb\_deliver} $m$.
\label{th:safety}
\end{theorem}
\begin{proof}
If the senders of $m$ and $m'$ are the same, then sequence number of $m$ is less than sequence number of $m'$ (follows from Definition~\ref{def:hbbyz}).
It is straightforward to observe from the algorithm that messages are BCRB-readyed and hence BCRB-delivered in source-sequence number order at $p_l$. So in this case, the theorem follows.

If the senders are different, then 
there may exist a transitive sequence of messages $\langle m= m_1, m_2, \ldots, m_z = m' \rangle$ from $m$ to $m'$ that are related by: 
\texttt{plaintext\_ready}($m_{x}$) before \texttt{bcrb\_\-broadcast}($m_{x+1}$) at process $p_{i_x}$, type of orderings.
So it suffices to show, in general, that: if $p_i$ (whether correct or Byzantine) \texttt{plaintext\_ready}($m_1$) before \texttt{bcrb\_\-broadcast}($m_2$) then no correct process triggers \texttt{bcrb\_\-deliver} $m_2$ before triggering \texttt{bcrb\_\-deliver} $m_1$.

Let $m_1 \gets (p_j, seq_j)$. Let \texttt{plaintext\_ready}($m_1$) occur at $p_i$ at physical time $t_1$. Before $t_1$, $\geq n-f$ ACK($m_1$) were sent, at least $n-2f$ were broadcast by correct $p_k$. Let these $n-2f$ form quorum $U_1$. Such $p_k$ must have received $2f+1$ READY($m_1$) and BCRB-readyed $m_1$; hence $LD_k[j] \geq seq_j$ at that time. 

It is a well-known property of Bracha's BRB that before the first correct process broadcasts READY($m$), it must have received $\frac{n+f}{2}$ ECHO($m$). This property holds even in our algorithm because the modifications we introduce only conditionally delay the processing of received ECHOs and READYs.

After $t_1$, before the first correct $p_l$ can broadcast READY($m_2$), it must have received at least $\frac{n+f}{2}$ ECHO($m_2$), $\ge \frac{n-f}{2}$ of which were from correct processes and all of which were sent after $t_1$. Let these $\frac{n-f}{2}$ form quorum $U_2$.
Denote this condition leading to the definition of quorum $U_2$ as Case-I.

The sum $|U_1| + |U_2| = n-2f + \frac{n-f}{2} = \frac{3n-5f}{2} > n-f$, the number of correct processes. Hence and as $n > 3f$, $|U_1 \cap U_2| \geq 1$. Let any correct process in this intersection be denoted $p_z$. Clearly, $p_z$ broadcast ACK($m_1$) (before $t_1$) which is before it broadcast ECHO($m_2$) (after $t_1$). From Observation~\ref{obs:ackDcbLD} and FIFO channels, $D_l[j,z] \ge seq_j$ at $p_l$ on receiving this ACK($m_1$) which happens before receiving ECHO($m_2$) from $p_z$. On receiving this $(M' \gets)$ ECHO($m_2$), $p_l$ set $M'.CB[j] \gets D[j,z] \ge seq_j$ and $M'.from \gets z$. Clearing the causal barrier of $M'$ requires $LD_l[j] \ge M'.CB[j] \ge seq_j$ and that can happen only after $LD_l[j] \gets seq_j$. (This is bound to happen as $m_1$ will certainly be BCRB-readyed at $p_l$, by Lemma~\ref{lemma:agree}.)  Further, when $LD_l[j] \gets seq_j$, $p_l$ enqueued  $m_1$ in $bcrb\_Q$ for BCRB\_delivery (possibly pending decryption) and broadcast ACK($m_1$). Clearing the causal barrier of $M'$ is a prerequisite to broadcasting READY($m_2$). So $p_l$ enqueues $m_1$ and broadcasts ACK($m_1$) before it broadcasts READY($m_2$) and clearly before it can enqueue $m_2$. 

Alternatively, from Bracha's algorithm, before $p_l$ can broadcast READY($m_2$), it should have received $f+1$ READY($m_2$), of which $\ge 1$ READY($m_2$) must be from correct $p_{l'}$. Our algorithm also requires that the causal barriers of $f+1$ READY($m_2$) be satisfied before $p_l$ broadcasts READY($m_2$). Denote this alternative condition for broadcasting READY($m$) as Case-II.

Let $p_{l^x}$ denote the $x$th $p_l$ to broadcast READY($m_2$). We show the following invariant by induction:
\begin{itemize}
    \item (Invariant/Induction hypothesis:) $LD_{l^x}[j] \geq seq_j$ and $p_{l^x}$ has enqueued $m_1$ in $bcrb\_Q$ and broadcast ACK($m_1$) before $p_{l^x}$ broadcasts READY($m_2$). 
\end{itemize}
For $x=1$, Case-I holds and $LD_{l^x}[j] \geq seq_j$ and the rest of the Invariant holds as shown for Case-I. For $x>1$, either case applies. If the first case applies, clearly using the argument given for it, $LD_{l^x}[j] \geq seq_j$ and the invariant holds. For Case-II, if $p_{l'}$ was $p_{l^y}$, then $y<x$. By I.H., $LD_{l^y}[j] = LD_{l'}[j] \geq seq_j$ and the invariant holds for $p_{l'}$. Then
using identical reasoning as for Case-I but replacing $p_z$ by $p_{l'}$ and ECHO($m_2$) by READY($m_2$) and not referencing $t_1$, the invariant can be seen to hold for $p_{l^x}$.
%
%
Thus in Case-II, $p_l$ enqueues $m_1$ in $bcrb\_Q$ before it clears the causal barriers of $f+1$ READY($m_2$) that includes the one from $p_{l'}$, and broadcasts READY($m_2$). Once that happens and $2f+1$ READY($m_2$) have been received, $p_l$ enqueues $m_2$ in $bcrb\_Q$. 

So combining both cases, a correct process enqueues $m_1$ in $bcrb\_Q$ and BCRB-readys $m_1$ before it broadcasts READY($m_2$); once it has broadcast READY($m_2$) and received $2f+1$ READY($m_2$), it enqueues $m_2$ in $bcrb\_Q$ and BCRB-readys $m_2$ (after $m_1$). As the algorithm triggers \texttt{bcrb\_deliver} in FIFO order dequeuing from $bcrb\_Q$, the theorem follows.
\end{proof}

\begin{corollary}
\label{cor:ordering}
If $m_1 \to m_2$ then at any correct process $p_i$, $m_1$ is BCRB-ready before $m_2$ is BCRB-ready.
\end{corollary}

\subsection{Validity, Agreement, and Integrity}
\label{sec:vai}
\begin{theorem}[Validity]
\label{th:validity}
If a correct process $\texttt{bcrb-broadcast}$s $m_1$, all correct processes $\texttt{bcrb-deliver}$ $m_1$.
\end{theorem}
\begin{proof}
When correct $p_i$ does $\texttt{bcrb-broadcast}(m_1)$, all correct processes $p_j$ receive INIT($m_1$) and broadcast ECHO($m_1$). 
ECHO($m_1$) from all $n-f$ correct processes will eventually be processed by all correct $n-f$ $p_k$s because by Lemma~\ref{lemma:agree}, all prior messages $m'$ BCRB-readyed by any $p_j$ before it broadcast ECHO($m_1$) will eventually be BCRB-readyed by all $p_k$. Once that happens at any $p_k$ (for ECHO($m_1$) from $\frac{n-f}{2}$ $p_j$), ($\forall a$) $LD_k[a] \ge LD_j[a]$ at the time $p_j$ broadcast ECHO($m_1$), hence those ECHO($m_1$) from all those $p_j$ get unblocked and $p_k$ then broadcasts READY($m_1$). Alternately, $p_k$ may receive $f+1$ READY($m_1$)s which each become unblocked (i.e., $m_1.blocked=0$) sooner, in which case also, it broadcasts READY($m_1$). 
In any case, all correct $p_c$ would broadcast READY($m_1$) and receive $2f+1 \le n-f$ READY($m_1$), clearing the way for BCRB-ready of $m_1$ and enqueuing $m_1$ in $bcrb\_Q_c$. All the correct $p_c$ upon BCRB-ready of $m_1$ would also broadcast the decryption share on ACK($m_1$); on receipt of these $n-f$ (valid) shares, $m_1$ would be flagged ready for BCRB-delivery. Entries $m^*$ ahead of it in the queue (which would have already been BCRB-readyed), whether from correct or Byzantine processes, would also be flagged ready eventually. This is because, by Lemma~\ref{lemma:agree}, $m^*$ would also have been BCRB-readyed at all $n-f$ correct processes and which would then broadcast ACK($m^*$) with their decryption share, clearing the way for flagging $m^*$ in $brcb\_Q$. Flagged entries are dequeued from head of $bcrb\_Q$ and $\texttt{bcrb\_deliver}$ is triggered. The theorem follows.
\end{proof}

\begin{theorem}[Agreement]
If a correct process $p_i$ BCRB-delivers $M$, every correct process also BCRB-delivers $M$.
\label{th:agree}
\end{theorem}
\begin{proof}
  When correct $p_i$ BCRB-delivers $m$, it must have already BCRB-readyed $m$ and enqueued it in $bcrb\_Q$. By Lemma~\ref{lemma:agree}, all correct $p_c$ also BCRB-ready $m$ and enqueue it in $bcrb\_Q$. Then all correct $p_c$ broadcast their decryption share on ACK($m$); on receipt of these (valid) shares, $m$ would be flagged ready for BCRB-delivery at all correct $p_c$. Entries $m^*$ ahead of it in the queue, whether from correct or Byzantine processes, would also be flagged ready eventually. This is because, by Lemma~\ref{lemma:agree}, $m^*$ would also be (eventually) BCRB-readyed at all $n-f$ correct processes which would then broadcast ACK($m^*$) with their decryption share, clearing the way for flagging $m^*$ in $brcb\_Q$ locally. Flagged entries are dequeued from head of $bcrb\_Q$ and $\texttt{bcrb\_deliver}$ is triggered. The theorem follows. 
\end{proof}

\begin{theorem}[Integrity]
\label{th:integrity}
A message $m$ is BCRB-delivered at most once by each correct process, and if the sender is correct, then only if it was BCRB-broadcasted by the sender.
\end{theorem}
\begin{proof}
At-most-once delivery is guaranteed because $\mathbf{LD}_i[j]$ strictly increments upon delivery as messages are delivered in source-sequence order, preventing any duplicate processing. Authenticity for correct senders follows because correct processes only send messages via \texttt{bcrb\_broadcast}, and the system guarantees sender authenticity; a Byzantine process cannot forge messages from a correct sender due to authenticated channels and at least one correct process needs to be tricked into processing an INIT from a forged sender in order for the algorithm to proceed to BCRB-delivery. Even otherwise, enough valid decryption shares will not be generated for a forged sender's message.
\end{proof}

Byzantine-Induced Progress Halting (Accepted): 
A Byzantine sender can broadcast conflicting 
messages or messages with arbitrary sequence numbers or omit sending ECHOs/READYs/ACKs selectively. This can at worst result in the Byzantine sender’s own message stream being blocked or halted permanently at correct processes. 

\subsection{Complexity}
\label{sec:complexity}
For message $m$ in Bracha's 3-phase BRB, there are up to $n$ messages in the \texttt{init} phase, and up to $n^2$ messages in each of the \texttt{echo} and \texttt{ready} phases. All these messages are of size $\mathcal{O}(|m|)$. The fourth phase of our algorithm is the \texttt{ACK} phase with up to $n^2$ messages of size $\mathcal{O}(|m|)$. Thus: (1) message complexity is $\mathcal{O}(n^2)$, (2) message metadata size is $\mathcal{O}(1)$, and (3) communication word complexity is $\mathcal{O}(n^2(|m| + 1))$ which scales as $\mathcal{O}(n^2)$.

\section{Conclusions}
\label{sec:concl}
We proposed the first $\mathcal{O}(1)$ metadata overhead BCRB algorithm with $\mathcal{O}(n^2)$ communication word  complexity that guarantees strong safety and liveness. This is a solution having optimal metadata overhead.  Existing algorithms have $\mathcal{O}(n^3)$ communication word complexity \cite{Cachin2001,DBLP:journals/tcs/AuvolatFRT21}. 
Unlike prior protocols, our protocol integrates causal ordering logic with reliable delivery by reengineering Bracha's Byzantine Reliable Broadcast (BRB) layer \cite{DBLP:journals/iandc/Bracha87} requiring $\mathcal{O}(n^2)$ messages having $\mathcal{O}(1)$ metadata. Our protocol design, though specific to Bracha's algorithm because it relies on Bracha's specific quorum amplification and threshold properties, can also be analogously integrated into other BRB protocols, e.g., the Imbs-Raynal protocol \cite{DBLP:journals/ppl/ImbsR16} to solve BCRB.

\sloppy

\balance
\bibliographystyle{ACM-Reference-Format}
\bibliography{references}


\end{document}